\documentclass[11pt,letterpaper,onecolumn]{quantumarticle}
\pdfoutput=1

\usepackage{booktabs}
\usepackage{feyn}
\usepackage{graphicx}
\usepackage{latexsym}
\usepackage{amssymb}
\usepackage{amsmath}
\usepackage{amsthm}
\usepackage{amsfonts}
\usepackage{bm}
\usepackage{multirow}
\usepackage{color}
\usepackage{comment}
\usepackage{bbm}
\usepackage{marvosym}
\usepackage{wasysym}
\usepackage{enumerate}  
\usepackage{hyperref}
\usepackage{url}
\usepackage{tikz}
\usetikzlibrary{quantikz2}
\usepackage{algorithm}% http://ctan.org/pkg/algorithm
\usepackage{algpseudocode}
\usepackage[numbers,sort&compress]{natbib}
\usepackage{orcidlink}
\usepackage{CJKutf8}

\newcommand{\ii}{\mathrm{i}}

\newtheorem{theorem}{Theorem}
\newtheorem{corollary}{Corollary}
\newtheorem{lemma}{Lemma}
\newtheorem{definition}{Definition}
\newtheorem{proposition}{Proposition}

\newtheorem{open}{Open Problem}

\newcommand{\xCornell}{Department of Physics, Cornell University, Ithaca, NY, USA}

\begin{document}

\title{Controlled jump in the Clifford hierarchy}

\author{\begin{CJK*}{UTF8}{gbsn}Yichen Xu (许轶臣)~\orcidlink{0000-0001-7914-9438}\end{CJK*}}
\email{yx639@cornell.edu}
\affiliation{\xCornell}

\author{\begin{CJK*}{UTF8}{gbsn}Xiao Wang (王骁)~\orcidlink{0000-0003-2898-3355}\end{CJK*}}
% \author{Xiao Wang~\orcidlink{0000-0003-2898-3355}}
\affiliation{\xCornell}

\begin{abstract}
We develop a simple and systematic route to higher levels of the qubit Clifford hierarchy by coherently controlling Clifford operations. Our approach is based on Pauli periodicity, defined for a Clifford unitary $U$ as the smallest integer $m\ge 1$ such that $U^{2^{m}}$ is a Pauli operator up to phase. We prove a sharp controlled-jump rule showing that the controlled gate $CU$ lies strictly in level $m+2$ of the hierarchy, and equivalently that $CU$ lies in level $k$ if $U^{2^{k-2}}$ is Pauli while no smaller positive power of $U$ is Pauli. We further quantify the resources required to realize large level jumps in the Clifford hierarchy by proving an essentially tight upper bound on Pauli periodicity as a function of the number of qubits, which implies that accessing high hierarchy levels through controlled Cliffords requires a number of target qubits that grows exponentially with the desired level. We complement this limitation with explicit infinite families of Pauli-periodic Cliffords whose controlled versions achieve asymptotically optimal jumps. As an application, we propose a protocol for preparing logical catalyst states that enable logical $Z^{1/2^k}$ phase gates via phase kickback from a single jumped Clifford.
\end{abstract}

\maketitle

\tableofcontents

\section{Introduction}

Fault-tolerant quantum computation naturally divides into two regimes: operations that are protected by stabilizer structure alone, and those that require additional non-stabilizer resources. Clifford circuits occupy the first regime: they preserve Pauli operators under conjugation, admit efficient classical simulation, and arise as the default logical gates for large classes of stabilizer codes. Universality is achieved by supplementing them with state injection and adaptive Clifford processing. The qubit Clifford hierarchy, introduced by Gottesman and Chuang~\cite{GottesmanChuang1999}, gives a precise stratification of this landscape as a nested family $\mathcal{C}_1 \subset \mathcal{C}_2 \subset \cdots$ defined by the action on Paulis: $\mathcal{C}_1$ is the Pauli group, $\mathcal{C}_2$ is the Clifford group, and higher levels contain structured non-Clifford operations implementable by gate teleportation given suitable resource states~\cite{GottesmanChuang1999,BravyiKitaev2005,CampbellTerhalVuillot2017}. 

This hierarchy is not merely a definition, but a guide for where overhead enters fault-tolerant implementations: non-Clifford gates are indispensable for universality and typically dominate cost~\cite{BravyiKitaev2005,bravyi2012Phys.Rev.A,duclos-cianci2013Phys.Rev.A,fowler2012Phys.Rev.A,jones2013Phys.Rev.A,CampbellAnwarBrowne2012}, yet while standard synthesis from a fixed universal set such as Clifford+$T$ is possible (approximately or exactly)~\cite{dawson2005solovay,nielsen2010quantum,Kliuchnikov2012Asymptotically,giles2013exact,RossSelinger2014,Mooney2021costoptimalsingle}, many useful primitives naturally sit above the third level (e.g., multi-controlled phases and fine-grained $Z$ rotations), so direct access to higher-level structure can yield savings in $T$ count, depth, and magic-state overhead. At the same time, transversal realizations of logical gates are fundamentally constrained by the Eastin--Knill theorem~\cite{EastinKnill2009}. Moreover,  locality-preserving logical gates in $D$ spatial dimensions are constrained within the $D$th level of the Clifford hierarchy~\cite{BravyiKoenig2013,PastawskiYoshida2015}. This motivates the construction of error correction codes that possess transversal gates in higher Clifford hierarchy~\cite{anderson2014,ha2018towards,hu2025climbing,golowich2025quantum}.  Meanwhile, the global structure of $\mathcal{C}_k$ on many qubits remains poorly understood, and most progress has come from isolating special families, namely semi-Cliffords~\cite{ZengChenChuang2008}, generalized semi-Cliffords and structural results for level three~\cite{BeigiShor2010,Pllaha2020unweylingclifford,desilva2021Proc.A,deSilvaLautsch2025}, diagonal hierarchy gates~\cite{CuiGottesmanKrishna2017,RengaswamyCalderbank2019}, and permutation gates at level three~\cite{he2025characterization}.
These analyses, while valuable, do not address a ubiquitous algorithmic mechanism: the addition of coherent control to an existing gate.

Controlled unitaries appear in phase estimation, conditional logic, arithmetic subroutines, and measurement-based gadgets. 
From a circuit viewpoint, adding a control qubit appears to be a mild modification, but its effect on hierarchy level can be dramatic and is far from automatic. 
Two natural questions arise for controlling a unitary $U$.
First, when does controlled-$U$ belong to the Clifford hierarchy at all?
Second, when it does, what is the \emph{exact} level containing it?
Anderson and Weippert derived strong necessary conditions for controlled gates to lie in the qubit Clifford hierarchy, providing evidence that controlled gates form a highly constrained subclass~\cite{AndersonWeippert2024}. 
In a complementary special case, Surti, Daguerre and Kim showed that if a Clifford squares to a Pauli string then its controlled version lies in the third level~\cite{surti2025efficient}. 
In parallel, Ref.~\cite{kim2025any} analyzed optimized depth for controlled-Clifford constructions within Clifford+$T$ compilation, emphasizing that synthesis costs and hierarchy level are related but distinct notions. 
What has been missing is a broad and sharp criterion that simultaneously answers both questions for controlled targets and quantifies the resources required to achieve large level jumps. 

In this work we initiate a systematic study of generating higher-level Clifford-hierarchy gates via controlled-Clifford operations. 
We introduce Pauli periodicity, which is defined as the least number of repeated squarings needed for a Clifford unitary to become a Pauli operator up to phase. 
Our main structural result gives an exact controlled jump rule: if a Clifford gate $U$ has Pauli periodicity $m$, then the controlled-$U$ gate lies strictly in the $(m+2)$th level of the hierarchy, meaning it is contained in $\mathcal{C}_{m+2}$ but not in $\mathcal{C}_{m+1}$. 
This turns the existence of a Pauli power, which previously appeared as a necessary condition, into a complete classification for controlled Cliffords. 
It also subsumes the known third-level characterization as the case $m=1$~\cite{surti2025efficient}. 

Having identified the exact hierarchy level of controlled-$U$ in terms of the Pauli periodicity of $U$, the natural next question is how large a Pauli periodicity, and hence how high a hierarchy level, can be attained with a register of $n$ qubits. 
Using the binary symplectic representation of Clifford operators, we prove a tight upper bound on Pauli periodicity in terms of the number of qubits. 
Combining this bound with the controlled jump rule yields an exponential lower bound on the number of target qubits required to realize controlled gates that first appear in high hierarchy levels. 
This explains why large hierarchy jumps cannot be obtained on small registers by simply adding controls.

We complement these limitations with explicit constructions that achieve the extreme behavior. 
We propose a family of periodic Clifford gates that saturate the qubit lower bound and therefore realize asymptotically optimal controlled jumps. 
We also show that the resulting jumped Cliffords admit exact decompositions over Clifford+$T$, which connects the hierarchy classification to practical compilation.

Finally, we connect the algebraic results to a fault-tolerant application in higher-order phase resources. 
Drawing on our algebraic results, we present a protocol for preparing a logical catalyst state that enables logical $Z^{1/2^k}$ phase gates. 
This provides a concrete route by which controlled-Clifford structure can be used to access fine-grained phase gates in a fault-tolerant fashion~\cite{GottesmanChuang1999,BravyiKitaev2005,CampbellTerhalVuillot2017}.

\section{Periodic Clifford gates and controlled jump}

\subsection{Preliminaries}
To set the stage for our discussion, we review some useful definitions and facts about the qubit Clifford hierarchy.

\begin{definition}[$n$-qubit Pauli group]
Let $n\ge 1$. Define the single-qubit Pauli operators
\begin{equation}
  X=\begin{pmatrix}0&1\\1&0\end{pmatrix},\qquad
  Y=\begin{pmatrix}0&-\ii\\ \ii&0\end{pmatrix},\qquad
  Z=\begin{pmatrix}1&0\\0&-1\end{pmatrix}.
\end{equation}
The $n$-qubit Pauli group is the subgroup $\mathcal{P}_n\subset U(2^n)$ generated by $\{\ii I, X_j, Y_j, Z_j\}_{j=1}^n$, where $X_j:=I^{\otimes (j-1)}\otimes X\otimes I^{\otimes (n-j)}$ (and similarly for $Y_j,Z_j$). Equivalently,
\begin{equation}
  \mathcal{P}_n=\bigl\{\omega\, P_1\otimes\cdots\otimes P_n:\; \omega\in\{\pm 1,\pm\ii\},\; P_j\in\{I,X,Y,Z\}\bigr\}.
\end{equation}
\end{definition}

\begin{definition}[Qubit Clifford hierarchy~\cite{GottesmanChuang1999}]\label{def:clifford-hierarchy}
Fix $n\ge 1$. Let $\mathcal{P}_n$ denote the $n$-qubit Pauli group. The $n$-qubit Clifford hierarchy is the nested family of subsets $\{\mathcal{C}_k^{(n)}\}_{k\ge 1}$ defined recursively by
\begin{equation}
  \mathcal{C}_1^{(n)}:=\mathcal{P}_n,\qquad
  \mathcal{C}_{k+1}^{(n)}:=\bigl\{U\in U(2^n):\; U P U^{\dagger}\in \mathcal{C}_k^{(n)}\ \text{for all}\ P\in \mathcal{P}_n\bigr\}.
\end{equation}
In particular, the second level $\mathcal{C}_2^{(n)}$ is the Clifford group. We denote the collection of all unitaries in the $n$ qubit hierarchy as $\mathcal{CH}$.
\end{definition}

\begin{proposition}[Basic properties of $\mathcal{CH}$]\label{prop:CH-basic}
Fix $n\ge 1$.
\begin{enumerate}[(1)]
  \item (\textit{Nestedness}) For all $k\ge 1$, one has $\mathcal{C}_k^{(n)}\subseteq \mathcal{C}_{k+1}^{(n)}$. For all $m\geq n$, $\mathcal{C}_k^{(n)}\subseteq \mathcal{C}_{k}^{(m)}$.
  \item (\textit{Group property}) The second level $\mathcal{C}_2^{(n)}$ is a group under multiplication (the Clifford group). For $k\ge 3$, $\mathcal{C}_k^{(n)}$ is in general not a group.
  \item (\textit{Clifford invariance}) If $U\in \mathcal{C}_k^{(n)}$ and $C_1,C_2\in \mathcal{C}_2^{(n)}$, then $C_1 U C_2\in \mathcal{C}_k^{(n)}$.
  \item (\textit{Generating set for $\mathcal{C}_2^{(n)}$}) The Clifford group $\mathcal{C}_2^{(n)}$ is generated by the $S$ gate
  \begin{equation}
    S:=\begin{pmatrix}1&0\\0&\ii\end{pmatrix}
  \end{equation}
  the Hadamard gate
  \begin{equation}
    H:=\frac{1}{\sqrt{2}}\begin{pmatrix}1&1\\1&-1\end{pmatrix}
  \end{equation}
  and the two-qubit controlled-NOT gate
  \begin{equation}
    \mathrm{CNOT}:=\ket{0}\!\bra{0}\otimes I+\ket{1}\!\bra{1}\otimes X.
  \end{equation}
  \item (Clifford+$T$ circuit) Let
  \begin{equation}
    T:=\mathrm{diag}(1,e^{\ii\pi/4})
  \end{equation}
  denote the single-qubit $T$ gate. For unitaries in levels $k\ge 3$, Clifford+$T$ is approximately universal in the sense that for any $V\in U(2^n)$ and any $\varepsilon>0$, there exists a Clifford+$T$ circuit that $\varepsilon$-approximates $V$ with depth $\mathrm{polylog}(1/\varepsilon)$ ~\cite{nielsen2010quantum,dawson2005solovay}.
  \item (Single-qubit phase gates in the hierarchy) For $k\ge 0$, the single-qubit phase gate
  \begin{equation}
    Z^{1/2^k}:=\mathrm{diag}\bigl(1,e^{\ii\pi/2^k}\bigr)
  \end{equation}
  lies in the $(k+1)$st level $\mathcal{C}_{k+1}^{(1)}$ (with the convention $\mathcal{C}_1^{(1)}=\mathcal{P}_1$)~\cite{Mooney2021costoptimalsingle,CuiGottesmanKrishna2017}.
\end{enumerate}
\end{proposition}

\begin{definition}[Controlled unitary]
Let $n\ge 1$ and let $U\in U(2^n)$ be an $n$-qubit unitary. The controlled-$U$ gate is the $(n+1)$-qubit unitary
\begin{equation}
  CU:=\ket{0}\!\bra{0}\otimes I_{2^n}+\ket{1}\!\bra{1}\otimes U\in U(2^{n+1}),
\end{equation}
where the first qubit is the control. Equivalently, in the computational basis ordered as $\{\ket{0}\otimes\ket{x},\ket{1}\otimes\ket{x}\}_{x\in\{0,1\}^n}$, $CU$ has the block-diagonal form
\begin{equation}
  CU=\begin{pmatrix}
    I_{2^n} & 0 \\
    0 & U
  \end{pmatrix}.
\end{equation}
\end{definition}

\begin{proposition}\label{prop:cpclif}
  Let $P\in \mathcal{P}_n$ be an $n$-qubit Pauli, then $CP\in \mathcal{C}_2^{(n+1)}$.  
\end{proposition}

\subsection{The controlled jump}

Our main study concerns a basic but subtle question: given an $n$-qubit unitary $U$, when does adding a control promote it to a higher level of the Clifford hierarchy? Concretely, we seek how the hierarchy level of $CU$ relates to that of $U$.

A general necessary condition for $CU$ to be in the Clifford hierarchy was proved by Anderson and Weippert~\cite{AndersonWeippert2024}.
\begin{theorem}[\cite{AndersonWeippert2024} Corollary 2.4.1]\label{thm:AW-necessary}
  A controlled unitary $CU$ is in $\mathcal{CH}$ only if $U\in \mathcal{CH}$ and $U^{2^m}=P$ for some $m\in \mathbb{N}$ and $P\in \mathcal{P}_n$ is some $n$-qubit Pauli.
\end{theorem}
\noindent The proof relies on a few elementary identities for controlled unitaries, which we record here for later use.
\begin{lemma}[Properties of controlled gates]
Let $U\in U(2^n)$ be a unitary on the $n$-qubit target register.
\begin{enumerate}[(1)]
  \item (Distribution law of controlled gates) Let $V\in U(2^n)$. Then
  \begin{equation}\label{eq:cdistribution}
    C(UV)=CU\, CV.
  \end{equation}
  Equivalently, in circuit form (with gates applied from left to right),
  \begin{equation}
    \begin{quantikz}[column sep=1.2em, row sep=1.0em]
      & \qw & \ctrl{1} & \qw \\
      & \qwbundle{n} & \gate{UV} & \qw
    \end{quantikz}
    =
    \begin{quantikz}[column sep=1.2em, row sep=1.0em]
      & \qw & \ctrl{1} & \ctrl{1} & \qw \\
      & \qwbundle{n} & \gate{V} & \gate{U} & \qw
    \end{quantikz}.
  \end{equation}
  \item (Control-$X$ conjugation) Let $X$ denote the Pauli-$X$ on the control qubit. Then
  \begin{equation}\label{eq:xcu}
    CU (X\otimes I_{2^n})
    = (I_2\otimes U^\dagger)\,C(U^2)\,(X\otimes I_{2^n})\,CU.
  \end{equation}
  This identity can be visualized as the circuit equivalence
  \begin{equation}
    \begin{quantikz}[column sep=1.2em, row sep=1.0em]
      & \qw & \gate{X} & \ctrl{1} & \qw \\
      & \qwbundle{n} & \qw & \gate{U} & \qw
    \end{quantikz}
    =
    \begin{quantikz}[column sep=1.2em, row sep=1.0em]
      & \qw & \ctrl{1} & \gate{X} & \ctrl{1} & \qw & \qw \\
      & \qwbundle{n} & \gate{U} & \qw & \gate{U^{2}} & \gate{U^{\dagger}} & \qw
    \end{quantikz}.
  \end{equation}

  \item (Pushing a target unitary through a control) For any $V\in U(2^n)$ acting on the target register,
  \begin{equation}\label{eq:vcu}
    CU(I_2\otimes V) = C(UVU^\dagger V^\dagger)\,(I_2\otimes V) \,CU.
  \end{equation}
  Equivalently, in the circuit diagram, we have
  \begin{equation}
    \begin{quantikz}[column sep=1.2em, row sep=1.0em]
      & \qw & \qw & \ctrl{1} & \qw \\
      & \qwbundle{n} & \gate{V} & \gate{U} & \qw
    \end{quantikz}
    =
    \begin{quantikz}[column sep=1.2em, row sep=1.0em]
      & \qw & \ctrl{1} & \qw & \ctrl{1} & \qw \\
      & \qwbundle{n} & \gate{U} & \gate{V} & \gate{UVU^{\dagger}V^{\dagger}} & \qw
    \end{quantikz}.
  \end{equation}
\end{enumerate}
\end{lemma}
The relations can be verified straightforwardly.

\begin{lemma}[Controlled-block-diagonal unitary]\label{lemma:blockdiagonal}
  Let $A,B\in U(2^n)$. For every $r\ge 1$, if a controlled-block-diagonal unitary
  \begin{equation}
    D=\ket{0}\!\bra{0}\otimes A+\ket{1}\!\bra{1}\otimes B
  \end{equation}
  lies in $\mathcal{C}_r^{(n+1)}$, then both $A$ and $B$ lie in $\mathcal{C}_r^{(n)}$. 
\end{lemma}

\begin{proof}
  For any $P\in\mathcal{P}_n$ one has
  \begin{equation}
    D\,(I_2\otimes P)\,D^{\dagger}=\ket{0}\!\bra{0}\otimes (APA^{\dagger})+\ket{1}\!\bra{1}\otimes (BPB^{\dagger}),
  \end{equation}
  which lies in $\mathcal{C}_{r-1}^{(n+1)}$ by definition. Applying the same reasoning recursively for $r-1,r-2,\ldots,1$ shows that for each $P$ the operators $APA^{\dagger}$ and $BPB^{\dagger}$ eventually land in $\mathcal{P}_n$ after $r-1$ steps, i.e., $A,B\in\mathcal{CH}$. More precisely, $A,B\in\mathcal{C}_r^{(n)}$.
\end{proof}

\begin{proof}[Proof of Theorem~\ref{thm:AW-necessary}]
  We argue by iterating the commutator-type identity~\eqref{eq:xcu}. Fix $k\ge 1$ and suppose that $CU\in \mathcal{C}_{k}^{(n+1)}$.

  Let $X$ denote the Pauli-$X$ on the control qubit and set
  \begin{equation}
    W_1 :=(X\otimes I_{2^n})\, CU\,(X\otimes I_{2^n})\,CU^\dagger.
  \end{equation}
  By Definition~\ref{def:clifford-hierarchy}, conjugation by an element of $\mathcal{C}_{k}^{(n+1)}$ maps Paulis to $\mathcal{C}_{k-1}^{(n+1)}$. Hence $W_1\in \mathcal{C}_{k-1}^{(n+1)}$. More generally, define recursively for $m\ge 1$
  \begin{equation}
    W_{m+1}:=(X\otimes I_{2^n})\,W_m\,(X\otimes I_{2^n})\,W_m^\dagger.
  \end{equation}
  Using again the defining property of the hierarchy together with the Clifford-invariance property (Proposition~\ref{prop:CH-basic}(3)), we obtain the level bound
  \begin{equation}
    W_m\in \mathcal{C}_{k-m}^{(n+1)}\qquad \text{for all }1\le m\le k-1.
  \end{equation}

  On the other hand, Eq.~\eqref{eq:xcu} implies by a straightforward induction that $W_m$ is diagonal in the control qubit and equals
  \begin{equation}\label{eq:wm}
    W_m=
    \begin{pmatrix}
      U^{-2^{m-1}} & 0\\
      0 & U^{2^{m-1}}
    \end{pmatrix}.
  \end{equation}
  In particular, taking $m=k-1$ gives $W_{k-1}\in \mathcal{C}_1^{(n+1)}=\mathcal{P}_{n+1}$. Since $W_{k-1}$ has the above block-diagonal form, this forces $U^{2^{k-2}}$ to be an $n$-qubit Pauli (up to an overall phase), i.e., $U^{2^{k-2}}=\pm U^{2^{k-2}}=P$ for some $P\in \mathcal{P}_n$.

  It remains to justify the condition $U\in\mathcal{CH}$. Fix any target Pauli $P\in\mathcal{P}_n$. Since $I_2\otimes P\in\mathcal{P}_{n+1}$ and $CU\in\mathcal{C}_k^{(n+1)}$, we have
  \begin{equation}\label{eq:cu-conj-target-pauli}
    CU\,(I_2\otimes P)\,CU^{\dagger}\in\mathcal{C}_{k-1}^{(n+1)}.
  \end{equation}
  But $CU$ is block diagonal in the control qubit, so the conjugate is explicitly
  \begin{equation}
    CU\,(I_2\otimes P)\,CU^{\dagger}=\ket{0}\!\bra{0}\otimes P+\ket{1}\!\bra{1}\otimes (UPU^{\dagger}).
  \end{equation}
  Applying Lemma~\ref{lemma:blockdiagonal} to the block-diagonal operator in Eq.~\eqref{eq:cu-conj-target-pauli} yields $UPU^{\dagger}\in\mathcal{C}_{k-1}^{(n)}$ for all Paulis $P\in\mathcal{P}_n$. Therefore $U\in\mathcal{C}_k^{(n)}\subseteq\mathcal{CH}$.
\end{proof}

While Theorem~\ref{thm:AW-necessary} constrains when $CU$ can lie in $\mathcal{CH}$, it does not by itself identify the hierarchy level of $CU$. Indeed, the above argument only shows that some power $U^{2^{m}}$ must eventually fall into the Pauli group; because the hierarchy is nested, this can happen strictly before the iteration reaches $m=k-1$, and therefore the bound obtained from the proof need not be tight.

A complementary characterization in the Clifford case $U\in\mathcal{C}_2^{(n)}$ was obtained by Surti, Daguerre, and Kim~\cite{surti2025efficient}:
\begin{theorem}[\cite{surti2025efficient} Lemma 1]\label{thm:surti}
  Let $U\in \mathcal{C}_2^{(n)}$ be a Clifford unitary. Then $CU\in \mathcal{C}_3^{(n+1)}$ if and only if $U^2=P$ for some Pauli $P\in \mathcal{P}_n$.
\end{theorem}

\begin{proof}
  ($CU\in\mathcal{C}_3^{(n+1)}\implies  U^2=P$) This follows from Theorem~\ref{thm:AW-necessary}.

  ($U^2=P\implies CU\in \mathcal{C}_3^{(n+1)}$) Assume $U\in\mathcal{C}_2^{(n)}$ and $U^2\in\mathcal{P}_n$. To show $CU\in\mathcal{C}_3^{(n+1)}$, it suffices by Definition~\ref{def:clifford-hierarchy} to verify that for every Pauli operator $Q\in\mathcal{P}_{n+1}$, the conjugate $CU\,Q\,CU^{\dagger}$ is a Clifford, i.e., lies in $\mathcal{C}_2^{(n+1)}$.

  Every $(n+1)$-qubit Pauli can be written (up to phase) as a product of a Pauli on the control and a Pauli on the target, so it is enough to check generators $X\otimes I_{2^n}$, $Z\otimes I_{2^n}$, and $I_2\otimes P$ with $P\in\mathcal{P}_n$.

  First, $Z\otimes I_{2^n}$ commutes with $CU$.

  Next, applying Eq.~\eqref{eq:xcu} and using $U^{-2}=(U^2)^{\dagger}\in\mathcal{P}_n$, we obtain
  \begin{equation}
    CU\,(X\otimes I_{2^n})\,CU^{\dagger}
    =(X\otimes I_{2^n})\,(I_2\otimes U^{-2})\,C(U^2).
  \end{equation}
  Since $U^2$ is Pauli, the controlled-Pauli $C(U^2)$ is a Clifford (it is a multi-controlled Pauli phase/bit-flip), and $I_2\otimes U^{-2}$ is also Pauli; therefore the right-hand side is a product of Cliffords and hence lies in $\mathcal{C}_2^{(n+1)}$.

  Finally, for any $P\in\mathcal{P}_n$, Eq.~\eqref{eq:vcu} gives
  \begin{equation}\label{eq:cui2p}
    CU\,(I_2\otimes P)\,CU^{\dagger}=C(UPU^{\dagger}P^{\dagger})\,(I_2\otimes P).
  \end{equation}
  Because $U$ is Clifford, $UPU^{\dagger}\in\mathcal{P}_n$, hence the commutator $UPU^{\dagger}P^{\dagger}$ is again Pauli, so $C(UPU^{\dagger}P^{\dagger})\in\mathcal{C}_2^{(n+1)}$ from Proposition~\ref{prop:cpclif}. Multiplying by $I_2\otimes P\in\mathcal{C}_1^{(n+1)}$ preserves membership in $\mathcal{C}_2^{(n+1)}$.

  This proves that $CU$ conjugates $\mathcal{P}_{n+1}$ into $\mathcal{C}_2^{(n+1)}$, and therefore $CU\in\mathcal{C}_3^{(n+1)}$.
\end{proof}

Our goal is to extend this result to controlled jumps by an arbitrary number of hierarchy levels and to determine the exact hierarchy level of the controlled gate. To do so, it is convenient to record the precise moment at which a unitary becomes Pauli under repeated squaring.

\begin{definition}[Pauli periodicity]\label{def:pauli-periodicity}
  Let $U\in U(2^n)$. We say that $U$ is $m$-Pauli-periodic if
  \begin{equation}
    U^{2^m}\in \mathcal{P}_n\qquad\text{and}\qquad U^{2^t}\notin\mathcal{P}_n\ \text{for all }t<m,
  \end{equation}
  where membership in $\mathcal{P}_n$ is understood up to an overall phase $\{\pm 1,\pm\ii\}$.
\end{definition}
Equivalently, we can define the Pauli periodicity as the least number of times that a unitary $U$ has to be squared in order to become a Pauli operator:
\begin{equation}
  m=\min\bigl\{t\ge 0:\; U^{2^t}\in\mathcal{P}_n\bigr\}.
\end{equation}
Note that, even though every Pauli matrix squares to the identity, the actual order of an $m$-Pauli-periodic unitary $U$ can be $\mathrm{ord}(U)=2^{m+q}$, where $q=0,1,2$ is the $\log_2$-periodicity of the Pauli operator $U^{2^m}$. This depends on whether the Pauli operator $U^{2^m}$ is identity, Pauli string with $\pm 1$ phase, or Pauli string with $\pm \ii$ phase, respectively.

With this notion, we can state the exact hierarchy level of the controlled unitary.
\begin{theorem}[Controlled jump criterion]\label{thm:controlled-jump-criterion}
  Let $U\in \mathcal{C}_2^{(n)}$ be a Clifford unitary. Then the controlled gate $CU$ lies strictly in level $m+2$ if $U$ has Pauli periodicity $m$. That is, $CU\in\mathcal{C}_{m+2}^{(n+1)}\setminus\mathcal{C}_{m+1}^{(n+1)}$.
\end{theorem}
\noindent Theorem~\ref{thm:controlled-jump-criterion} is a refinement of the necessary condition in Theorem~\ref{thm:AW-necessary} and a generalization of Theorem~\ref{thm:surti}: here we not only require that some power of $U$ lies in the Pauli group, but show that the minimal such power precisely determines the hierarchy level of $CU$.
As we will see in the proof below, the assumption $U\in\mathcal{C}_2^{(n)}$ is essential here, without which we cannot pinpoint the level of $CU$.

\begin{proof}
We prove the membership and strictness statements separately.

\smallskip
\noindent Step 1: $CU\in\mathcal{C}_{m+2}^{(n+1)}$.
By Definition~\ref{def:clifford-hierarchy}, it suffices to show that for every $(n+1)$-qubit Pauli $Q\in\mathcal{P}_{n+1}$, the conjugate $CU\,Q\,CU^{\dagger}$ lies in $\mathcal{C}_{m+1}^{(n+1)}$. We prove this via induction.

First of all, the statement is clearly true for $m=1$ from the Theorem~\ref{thm:surti}. Now assuming the statement is true up to $m-1$. That is, for every $V$ that has Pauli periodicity $m'<m$, $CV$ lies strictly in level $m'+2$. We now verify the claim for $CU$ where $U$ has Pauli periodicity $m$. To this end, we note that every $(n+1)$-qubit Pauli is (up to phase) one of
\begin{equation}
  Z\otimes I_{2^n},\qquad I_2\otimes P,\qquad X\otimes P,\qquad (XZ)\otimes P,
\end{equation}
with $P\in\mathcal{P}_n$. We handle these cases one by one.
\begin{itemize}
  \item $Z\otimes I_{2^n}$ commutes with $CU$.

  \item For $I_2\otimes P$, since $U$ is Clifford we have $UPU^{\dagger}\in\mathcal{P}_n$, and Eq.~\eqref{eq:cui2p} shows $CU\,(I_2\otimes P)\,CU^{\dagger}$ is a product of a controlled-Pauli and a Pauli; hence it is a Clifford, i.e., it lies in $\mathcal{C}_2^{(n+1)}\subseteq\mathcal{C}_{m+1}^{(n+1)}$.

  \item For $X\otimes I_{2^n}$, Eq.~\eqref{eq:xcu} implies
  \begin{equation}\label{eq:conj-x}
    CU\,(X\otimes I_{2^n})\,CU^{\dagger}
    =(X\otimes I_{2^n})\,(I_2\otimes U^{-2})\,C(U^2).
  \end{equation}
  Because $U$ is Clifford, $U^{-2}$ is also Clifford, so $(X\otimes I_{2^n})(I_2\otimes U^{-2})\in\mathcal{C}_2^{(n+1)}$ from group property of Cliffords. By Proposition~\ref{prop:CH-basic}(3), the hierarchy level of the right-hand side is the same as that of $C(U^2)$.

  Set $V:=U^2$. Since $U$ has Pauli periodicity $m$, $V$ has Pauli periodicity $m-1$. By the induction hypothesis applied to $V$, we have
  \begin{equation}
    CV\in\mathcal{C}_{(m-1)+2}^{(n+1)}=\mathcal{C}_{m+1}^{(n+1)}.
  \end{equation}
  Hence $C(U^2)=CV\in\mathcal{C}_{m+1}^{(n+1)}$, and therefore Eq.~\eqref{eq:conj-x} shows $CU\,(X\otimes I)\,CU^{\dagger}$ is in the level $m+1$.

  \item For $X\otimes P$, we write it as $X\otimes P=(I_2\otimes P)(X\otimes I_{2^n})$. Then, using Eq.~\eqref{eq:cdistribution}, we have
  \begin{equation}
    CU\,(X\otimes P)\,CU^{\dagger}=
    \bigl(CU\,(I_2\otimes P)\,CU^{\dagger}\bigr)
    \bigl(CU\,(X\otimes I_{2^n})\,CU^{\dagger}\bigr),
  \end{equation}
  which is a product of a Clifford gate (from Eq.~\eqref{eq:cui2p}) and an element of $\mathcal{C}_{m+1}^{(n+1)}$. Hence, from Proposition~\ref{prop:CH-basic}(3), it lies in $\mathcal{C}_{m+1}^{(n+1)}$ as well.

  \item The case $(XZ)\otimes P$ follows similarly since $CU$ commutes with $Z\otimes I$.
\end{itemize}
All of these results prove that $CU\in\mathcal{C}_{m+2}^{(n+1)}$.

\smallskip
\noindent Step 2: $CU\notin\mathcal{C}_{m+1}^{(n+1)}$.
If $CU\in\mathcal{C}_{m+1}^{(n+1)}$, then the iteration in the proof of Theorem~\ref{thm:AW-necessary} (with $k=m+1$) would imply that $U^{2^{m-1}}$ is Pauli (up to phase), contradicting the assumption that $U$ has Pauli periodicity $m$.

Combining the two steps yields
$CU\in\mathcal{C}_{m+2}^{(n+1)}\setminus\mathcal{C}_{m+1}^{(n+1)}$.
\end{proof}

The proof above crucially uses Clifford invariance (Proposition~\ref{prop:CH-basic}(3)) together with the fact that Cliffords conjugate Paulis to Paulis. For general $U\in\mathcal{CH}$ (not necessarily Clifford), these conditions will fail, and the same argument does not determine the exact level of $CU$. We will briefly discuss controlled jumps from non-Clifford inputs in Sec.~\ref{sec:jumpfromnonclif}.

%\todo{Higher controls. Can we make statements?}

\subsection{Properties of the Jumped Cliffords}
Motivated by Theorem~\ref{thm:controlled-jump-criterion}, we isolate the class of controlled gates obtained from controlling Pauli-periodic Cliffords.

\begin{definition}[Jumped Cliffords]
  A jumped Clifford is a controlled unitary $CU$ where the target unitary $U\in\mathcal{C}_2^{(n)}$ is Pauli-periodic. 
  %We denote the set of all such controlled gates by $\mathcal{J}$.
\end{definition}

\begin{proposition}[Inverse of jumped Cliffords]
  The inverse of every jumped Clifford is again a jumped Clifford in the same level of the Clifford hierarchy.
\end{proposition}
This follows immediately from closure of the Clifford group under inversion, the distribution law for controlled gates, and the fact that the inverse of a unitary has the same periodicity as the original. This closure under inversion is a useful feature of jumped Cliffords. For $k\ge 3$, the higher levels $\mathcal{C}_k^{(n)}$ are not groups in general, so taking an inverse need not preserve the hierarchy level for an arbitrary $k$th-level unitary.

\begin{proposition}[Controlled jump of tensor product Cliffords]\label{prop:tensor-product-periodicity}
Let $U=U_1\otimes U_2$ where $U_i\in\mathcal{C}_2^{(n_i)}$ are Pauli-periodic Cliffords with Pauli periodicities $m_i$ (Definition~\ref{def:pauli-periodicity}). Then $U$ is Pauli-periodic with periodicity
\begin{equation}
  m=\max\{m_1,m_2\},
\end{equation}
and consequently the controlled gate $C(U_1\otimes U_2)$ lies in
\begin{equation}
  C(U_1\otimes U_2)\in \mathcal{C}_{m+2}^{(n_1+n_2+1)}.
\end{equation}
\end{proposition}
\noindent In fact, since $(U_1\otimes U_2)^{2^t}=U_1^{2^t}\otimes U_2^{2^t}$, the smallest $t$ for which this becomes a Pauli (up to phase) is exactly $t=\max\{m_1,m_2\}$, because $U_i^{2^t}$ is Pauli iff $t\ge m_i$.
With $m=\max\{m_1,m_2\}$, Theorem~\ref{thm:controlled-jump-criterion} applied to the Clifford $U_1\otimes U_2$ implies that $C(U_1\otimes U_2)\in\mathcal{C}_{m+2}$. Proposition~\ref{prop:tensor-product-periodicity} is meaningful in logical computation, where a transversal gate $\bar U$ is implemented by a tensor product of physical Cliffords acting on each individual party. It shows that the controlled logical gate $C(\bar{U})$ lies in the same level of the Clifford hierarchy as each physical controlled gate $C(U_p)$.

\begin{proposition}[Jumped Cliffords are Clifford+$T$]\label{prop:cliffordT}
  Every jumped Clifford admits an exact Clifford+$T$ decomposition.
\end{proposition}
\begin{proof}
  From Proposition~\ref{prop:CH-basic}(4), every Clifford $U$ can be expressed as a product of $S$ gates, Hadamards, and CNOTs. Using the distribution law in Eq.~\eqref{eq:cdistribution}, the corresponding controlled gate $CU$ can therefore be written as a product of gates of the form $C(S)$, $C(H)$, and $C(\mathrm{CNOT})$.
  Since each factor in such a product has an exact Clifford+$T$ decomposition~\cite{giles2013exact}, concatenating these decompositions yields an exact Clifford+$T$ circuit for $CU$.
\end{proof}

Proposition~\ref{prop:cliffordT} shows that every jumped Clifford can be implemented exactly using only Clifford+$T$ operations. Thus, unlike finer phase gates, even when a jumped Clifford lies in a high level of the Clifford hierarchy, it does not require approximate synthesis from a universal gate set.  This observation will be crucial in our logical phase gate protocol in Sec.~\ref{sec:qecapplication}.

\subsection{The qubit resource for controlled jumps}

Theorem~\ref{thm:controlled-jump-criterion} shows that achieving a jump to a high level requires a Clifford gate with large Pauli periodicity. Here we prove a general upper bound on how large the Pauli periodicity of an $n$-qubit Clifford can be, and we give an explicit family of Cliffords that attains this bound.

\begin{theorem}[Upper bound on Pauli periodicity]\label{thm:lowerboundqubit}
  Let $U\in\mathcal{C}_2^{(n)}$ be a Clifford unitary on $n$ qubits. Suppose that $U$ is $m$-Pauli-periodic (Definition~\ref{def:pauli-periodicity}), i.e., $U^{2^m}\in\mathcal{P}_n$ up to phase.
  Then
  \begin{equation}
    m\le \lceil\log_2(2n)\rceil.
  \end{equation}
  Moreover, the upper bound is tight: for each $n>1$ there exists a Clifford unitary $U\in\mathcal{C}_2^{(n)}$ whose Pauli periodicity is exactly $\lceil\log_2(2n)\rceil$.
\end{theorem}

Before proving Theorem~\ref{thm:lowerboundqubit}, we recall the standard binary symplectic description of $n$-qubit Clifford unitaries: a Clifford is determined (up to global phase) by its action on Pauli operators, which can be encoded by a symplectic matrix over $\mathbb{F}_2$ together with a linear phase function.
\begin{lemma}[Binary matrix representation of Cliffords~\cite{DD_2003,AaronsonGottesman2004}]\label{lem:binary-matrix-clifford}
  Let $U\in\mathcal{C}_2^{(n)}$ be an $n$-qubit Clifford unitary. Then there exist
  \begin{equation}
    F\in \mathrm{Sp}(2n,\mathbb{F}_2)\qquad\text{and}\qquad \gamma\in (\mathbb{Z}_4)^{2n}
  \end{equation}
  such that for every Pauli operator $P\in\mathcal{P}_n$ written (up to phase) as
  \begin{equation}
    P\equiv X^{\mathbf{x}}Z^{\mathbf{z}},\qquad (\mathbf{x},\mathbf{z})\in (\mathbb{F}_2)^{2n},
  \end{equation}
  one has
  \begin{equation}
    U\,P\,U^{\dagger} \equiv \ii^{\langle \gamma,(\mathbf{x},\mathbf{z})\rangle}\, X^{\mathbf{x}'}Z^{\mathbf{z}'},
    \qquad (\mathbf{x}',\mathbf{z}') = F(\mathbf{x},\mathbf{z}),
  \end{equation}
  where $\langle\cdot,\cdot\rangle$ denotes the natural pairing $(\mathbb{Z}_4)^{2n}\times(\mathbb{F}_2)^{2n}\to\mathbb{Z}_4$, and $\equiv$ denotes equality up to an overall phase.
  Conversely, any pair $(F,\gamma)$ of this form specifies a Clifford unitary up to global phase.

  %A common equivalent parametrization is by a binary symplectic matrix together with a linear (or affine) phase polynomial.
\end{lemma}

\begin{lemma}[Periodicity bound for invertible binary matrices]\label{lem:unipotent-symplectic}
  Let $F\in\mathrm{SL}(2n,\mathbb{F}_2)$ and suppose that $F$ has $2$-power order, i.e.,
  \begin{equation}
    F^{2^t}=I
  \end{equation}
  for some $t\ge 0$.
  Then $F$ is unipotent and the nilpotent matrix $N:=F-I$ satisfies
  \begin{equation}
    (F-I)^{2n}=N^{2n}=0.
  \end{equation}
  Moreover,
  \begin{equation}
    F^{2^{\lceil\log_2(2n)\rceil}}=I.
  \end{equation}
\end{lemma}

\begin{proof}
  Let $\overline{\mathbb{F}}_2$ be an algebraic closure of $\mathbb{F}_2$ and let $\lambda$ be an eigenvalue of $F$ over $\overline{\mathbb{F}}_2$.
  From $F^{2^t}=I$ we obtain $\lambda^{2^t}=1$.
  Since $\lambda$ lies in some finite extension $\mathbb{F}_{2^s}$, it belongs to the cyclic group $\mathbb{F}_{2^s}^{\times}$ of order $2^s-1$ (odd).
  Hence the only element of $2$-power order in $\mathbb{F}_{2^s}^{\times}$ is $1$, so $\lambda=1$.
  Thus all eigenvalues of $F$ are $1$, i.e., $F$ is unipotent.

  Since $F$ is unipotent, its minimal polynomial has the form $(x-1)^r$ with $1\le r\le 2n$.
  Equivalently, writing $F=I+N$ we have $N^r=0$, hence in particular $N^{2n}=0$.

  Over characteristic $2$, the Frobenius/binomial identity gives
  \begin{equation}\label{eq:frobenius-unipotent}
    (I+N)^{2^k}=I+N^{2^k}\qquad\text{for all }k\ge 0.
  \end{equation}
  Choosing $k=\lceil\log_2(2n)\rceil$ yields $2^k\ge 2n\ge r$, so $N^{2^k}=0$ and therefore $F^{2^k}=(I+N)^{2^k}=I$.
\end{proof}

\begin{proof}[Proof of Theorem~\ref{thm:lowerboundqubit}]
Let $F\in \mathrm{Sp}(2n,\mathbb{F}_2)$ be the binary symplectic matrix of $U$ from Lemma~\ref{lem:binary-matrix-clifford}.
If $U^{2^m}\in \mathcal{P}_n$ (up to phase), then conjugation by $U^{2^m}$ acts trivially on Pauli labels, hence
\begin{equation}\label{eq:pauli-implies-F-trivial}
  F^{2^m}=I.
\end{equation}
In particular, $F$ has $2$-power order, so Lemma~\ref{lem:unipotent-symplectic} applies and yields
\begin{equation}
  F^{2^{\lceil\log_2(2n)\rceil}}=I.
\end{equation}
Therefore $m\le \lceil \log_2(2n)\rceil$.

To show the upper bound is achievable, we use the existence of a regular unipotent element of large $2$-power order in $\mathrm{Sp}(2n,\mathbb{F}_2)$.
The group $\mathrm{Sp}(2n,\mathbb{F}_2)$ is of Lie type $C_n$ in defining characteristic $p=2$.
By Testerman's order formula for regular unipotent elements~\cite[Eq.~0.4]{Testerman1995A1},
there exists a regular unipotent element $x\in \mathrm{Sp}(2n,\mathbb{F}_2)$ whose order is the smallest power of the characteristic $p$ that is strictly larger than the height of the highest root.
For type $C_n$ the Coxeter number is $h=2n$, and the height of the highest root is $\mathrm{ht}(\alpha_0)=h-1=2n-1$~\cite{MITOCW18755Coxeter}.
Hence $\mathrm{ord}(x)$ is the smallest $2^k$ satisfying $2^k>2n-1$, i.e.,
\begin{equation}
  k=\lceil \log_2(2n)\rceil.
\end{equation}

Choose any Clifford unitary $U_x\in \mathcal{C}_2^{(n)}$ whose induced symplectic action is $F=x$ (such a Clifford exists by Lemma~\ref{lem:binary-matrix-clifford}).
Then $F^{2^k}=I$ implies $U_x^{2^k}\in \mathcal{P}_n$ up to phase.
Moreover, since $\mathrm{ord}(x)=2^k$ is exact, we have $F^{2^{k-1}}\neq I$, and therefore $U_x^{2^{k-1}}\notin \mathcal{P}_n$.
Thus $U_x$ has Pauli periodicity exactly $k=\lceil \log_2(2n)\rceil$.
\end{proof}

In Sec.~\ref{sec:sxcnoth}, we will provide a concrete example of a Pauli-periodic unitary that saturates the bound in Theorem~\ref{thm:lowerboundqubit}.

\begin{corollary}[Qubit lower bound for a $k$-level controlled jump]\label{cor:qubit-lb-kjump}
  Let $U\in\mathcal{C}_2^{(n)}$ be an $n$-qubit Clifford and suppose that $CU$ lies strictly in level $k\geq 4$, i.e.,
  $CU\in\mathcal{C}_k^{(n+1)}\setminus\mathcal{C}_{k-1}^{(n+1)}$.
  Then the target unitary $U$ must act nontrivially on at least $ n \;\ge\; 2^{k-4}+1$ qubits.

\end{corollary}
\begin{proof}
  If $CU$ is strictly in level $k$, then by Theorem~\ref{thm:controlled-jump-criterion} the Clifford $U$ has Pauli periodicity $m=k-2$.
  Applying Theorem~\ref{thm:lowerboundqubit} to $U$ gives
  \begin{equation}
    k-2=m \le \lceil\log_2(2n)\rceil.
  \end{equation}
  This implies $2n>2^{k-3}$ and hence the smallest possible integer $n$ is $2^{k-4}+1$ for $k\geq 4$.
\end{proof}

\section{Examples of Pauli-periodic and jumped Cliffords}

In this section, we provide a few examples of Pauli-periodic Cliffords and their jumped Cliffords. In particular, we construct permutation gates in higher Clifford hierarchy via controlled Clifford permutations. In addition, we give the example of an optimal periodic Clifford that saturates the lower bound in Corollary~\ref{cor:qubit-lb-kjump}.

\subsection{Controlled Clifford permutations}
\label{sec:cclifperm}
We first consider control gates of a simple and well-structured subclass of Clifford unitaries: Pauli-periodic Cliffords that act as permutations of computational basis states.

\begin{definition}[Permutation gate]
  An $n$-qubit unitary $U$ is called a permutation gate if there exists a permutation $\pi$ of $\{0,1\}^n$ such that
  \begin{equation}
    U_\pi\ket{a}=\ket{\pi(a)}\qquad \text{for all }a\in\{0,1\}^n.
  \end{equation}
\end{definition}
\begin{proposition}[Polynomial descriptions of a permutation gate]
  Every permutation gate over $n$ qubits can be written as
  \begin{equation}
    U_\pi=\sum_{a\in\{0,1\}^n}\ket{\pi(a)}\bra{a},
  \end{equation}
  where $\pi$ is a permutation of $\{0,1\}^n$. Writing $\pi(a)=(\pi_1(a),\ldots,\pi_n(a))$ and viewing
  $a=(a_1,\ldots,a_n)$ as a vector over $\mathbb{F}_2$, each coordinate function
  $\pi_i:\{0,1\}^n\to\{0,1\}$ admits a unique representation as a polynomial in $a_1,\ldots,a_n$
  over $\mathbb{F}_2$.
\end{proposition}

Since we are considering Pauli-periodic Cliffords, one interesting class is the Pauli-periodic Clifford permutation.
\begin{definition}[Clifford permutation]
  An $n$-qubit Clifford permutation is a permutation gate that can be generated by a sequence of CNOT and $X$ gates.
\end{definition}

For Clifford permutations, we have these useful results:
\begin{lemma}[Clifford permutations are affine linear~\cite{he2025characterization}, Proposition~2.15]\label{lem:clifford-perm-affine}
  Let $U$ be an $n$-qubit Clifford permutation. Then there exist a binary matrix
  $M\in \mathrm{GL}(n,\mathbb{F}_2)$ and a fixed bitstring $\phi\in\{0,1\}^n$ such that for all $a\in\{0,1\}^n$,
  \begin{equation}
    U\ket{a}=\ket{Ma+\phi},
  \end{equation}
  where addition and matrix multiplication are over $\mathbb{F}_2$ (identifying bitstrings with column vectors in $\mathbb{F}_2^n$).
\end{lemma}

Lemma~\ref{lem:clifford-perm-affine} shows that Clifford permutations are precisely reversible affine maps on bitstrings.
To analyze their Pauli periodicity, it is useful to write down the associated binary symplectic representation.

\begin{proposition}[Binary symplectic matrix of a Clifford permutation]\label{prop:clifford-perm-symplectic}
  Let $U$ be an $n$-qubit Clifford permutation with linear part $M\in\mathrm{GL}(n,\mathbb{F}_2)$ as in Lemma~\ref{lem:clifford-perm-affine}. Then the binary symplectic matrix
  $F\in\mathrm{Sp}(2n,\mathbb{F}_2)$ of $U$ (Lemma~\ref{lem:binary-matrix-clifford}) is block diagonal and equals
  \begin{equation}
    F=
    \begin{pmatrix}
      M & 0\\
      0 & \left(M^{-1}\right)^T
    \end{pmatrix}.
  \end{equation}
\end{proposition}
\begin{proof}
  We verify this by directly computing the action of $U$ on Pauli generators. Let $e_i$ be the $i$th standard basis vector.  For each $i$, the operator $X_i$ acts as $\ket{a}\mapsto\ket{a+e_i}$, so
  \begin{equation}
    U X_i U^\dagger \ket{Ma+\phi}
    = U X_i \ket{a}
    = U\ket{a+e_i}
    = \ket{Ma+Me_i+\phi}.
  \end{equation}
  Thus $U X_i U^\dagger = X^{Me_i}$ up to phase. Equivalently, on $X$-labels the induced linear map is $\mathbf{x}\mapsto M\mathbf{x}$.

  \smallskip
  \noindent\emph{$Z$-type Paulis.}
  For each $i$, $Z_i\ket{a}=(-1)^{e_i\cdot a}\ket{a}$, hence
  \begin{equation}
    U Z_i U^\dagger \ket{Ma+\phi}
    = U Z_i \ket{a}
    = (-1)^{e_i\cdot a}\ket{Ma+\phi}.
  \end{equation}
  Writing $a=M^{-1}(Ma+\phi-\phi)$ gives $e_i\cdot a = ((M^{-1})^T e_i)\cdot (Ma+\phi) + \text{(constant)}$, so
  $U Z_i U^\dagger$ equals $Z^{(M^{-1})^T e_i}$ up to an overall phase.
  Therefore on $Z$-labels the induced linear map is $\mathbf{z}\mapsto (M^{-1})^T\mathbf{z}$.

  \smallskip
  Since conjugation by $U$ sends $X$-type Paulis to $X$-type Paulis and $Z$-type Paulis to $Z$-type Paulis, there is no $X$--$Z$ mixing, so the symplectic matrix is block diagonal with blocks $M$ and $(M^{-1})^T$.
\end{proof}

\begin{corollary}[Pauli-periodicity of Clifford permutations]\label{cor:clifperiod}
For a Clifford permutation to be Pauli-periodic, the corresponding matrix $M$ must be unipotent. Furthermore, the maximum Pauli periodicity that can be reached by any Clifford permutation on $n$ qubits is upper bounded by $\lceil \log_2(n)\rceil$.
\end{corollary}
\noindent Indeed, if $U^{2^m}\in\mathcal{P}_n$ up to phase, then conjugation by $U^{2^m}$ acts trivially on Pauli labels, and hence the associated symplectic matrix satisfies $F^{2^m}=I$. Using Proposition~\ref{prop:clifford-perm-symplectic}, we obtain
$M^{2^m}=I$ and $\bigl((M^{-1})^T\bigr)^{2^m}=I$, and thus also $(M^T)^{2^m}=I$.
Therefore $M$ has $2$-power order, which (over $\mathbb{F}_2$) implies that $M$ is unipotent. Applying Lemma~\ref{lem:unipotent-symplectic} for the $n\times n$ matrix $M$ yields $m\le \lceil\log_2 n\rceil$, as claimed. Moreover, this upper bound can be achieved, for example, by taking $M$ to be a single Jordan block with ones on the diagonal and the superdiagonal, i.e.,
\begin{equation}\label{eq:jordan}
  M_{ij}=\begin{cases}
    1, & i=j \text{ or } i=j-1,\\
    0, & \text{otherwise}.
  \end{cases}
\end{equation}
The Clifford permutation corresponding to this $M$ is precisely the CNOT string $\left(\prod_{j=1}^{n-1}\mathrm{CNOT}_{j+1,j}\right)$ where the product is taken from right to left. Another simple choice that saturates the upper bound in Corollary~\ref{cor:clifperiod}, which can be realized in a constant depth circuit, is the brickwork CNOT circuit:
\begin{equation}\label{eq:brickworkcnot}
  \left(\prod_{j=1}^{\lfloor\frac{n-1}{2}\rfloor}\mathrm{CNOT}_{2j+1,2j}\right)\left(\prod_{j=1}^{\lfloor\frac{n}{2}\rfloor}\mathrm{CNOT}_{2j,2j-1}\right).
\end{equation}

We now consider the properties of the jumped Clifford permutations.

\begin{proposition}[Jumped Clifford permutations are quadratic]\label{prop:jumped-clifford-perm-lowdeg}
  Let $U$ be an $n$-qubit Clifford permutation, and let $CU$ denote the corresponding controlled gate on $n+1$ qubits.
  Then $CU$ is again a permutation gate. Moreover, identifying computational basis states with bitstrings
  $a=(a_0,a_1,\ldots,a_n)\in\{0,1\}^{n+1}$ (where $a_0$ is the control bit), the permutation induced by dressed permutation gate of the form $P\, CU\, Q$, where $P,Q$ are Clifford permutations over the $(n+1)$ qubits, can be written as
  \begin{equation}\label{eq:bigpi}
    a_i\longmapsto \Pi_i(a_0,a_1,\ldots,a_n),\qquad i=0,1,\ldots,n,
  \end{equation}
  where each $\Pi_i$ is a polynomial over $\mathbb{F}_2$ of degree at most $2$. 
\end{proposition}
\noindent This is straightforward to see: in a decomposition $P\,CU\,Q$, the Clifford permutation $Q$ first maps the input bits $a$ to new affine-linear bits $b=\Pi_Q(a)$, so that the control bit entering $CU$ is $b_0$. The only operation that can introduce products is the single controlled gate $CU$, which can only create cross terms involving $b_0$; composing with the subsequent Clifford permutation $P$ preserves the degree, so all output coordinates remain polynomials of total degree at most $2$ in the original input bits $a$. 

Proposition~\ref{prop:jumped-clifford-perm-lowdeg} also connects to the polynomial description of permutation gates. Lemma~2.13 of Ref.~\cite{he2025characterization} shows that, for permutation gates, the Clifford-hierarchy level upper-bounds the degree of a polynomial description of the inverse permutation $\pi^{-1}$. By contrast, the degree of the forward map $\pi$ itself does not directly control the hierarchy level. Jumped Clifford permutations provide a concrete illustration: even though they admit quadratic coordinate functions, they can still lie in arbitrarily high levels of the Clifford hierarchy given enough qubits.

Moreover, this hierarchy upper bound on $\deg(\pi^{-1})$ can be far from tight. Indeed, the inverse of a dressed jumped permutation is again of the form $Q^{\dagger}\,C(U^{\dagger})\,P^{\dagger}$, where $U^{\dagger}$ is a Clifford permutation with the same Pauli periodicity. Thus the inverse of the jumped permutation lives in the same level in the hierarchy. Meanwhile, the inverse map $\Pi^{-1}$ (with $\Pi$ defined in Eq.~\ref{eq:bigpi}) again admits a quadratic polynomial description, despite the high level nature of the original permutation gate.

In short, the controlled Clifford permutations demonstrate that low algebraic degree of $\pi$ or $\pi^{-1}$ does not preclude a permutation gate from belonging to high levels of the Clifford hierarchy.

\subsection{Optimal Jumped Cliffords: Controlled $S_X$-CNOT-$H$ Strings}\label{sec:sxcnoth}

We now construct another important example of Pauli-periodic Cliffords whose Pauli periodicity achieves the upper bound in Theorem~\ref{thm:lowerboundqubit}.

\begin{definition}[$S_X-CNOT-H$ string]
  Define $S_X:=HSH$ as the $\pi/2$ phase gate about the $X$ axis. For $n\ge 2$, we define the $S_X$-CNOT-$H$ string on $n$ qubits to be the Clifford unitary
  \begin{equation}
    SCH_n
    :=S_{X,n}\left(\prod_{j=1}^{n-1}\mathrm{CNOT}_{j,j+1}\right)H_1,
  \end{equation}
  where $\mathrm{CNOT}_{j,j+1}$ denotes a CNOT with control qubit $j$ and target qubit $(j+1)$, the product is taken from right to left (i.e. $S_{X,n}$ is followed by $CNOT_{n-1,n}$), and $S_{X,n}$ (resp.\ $H_1$) means applying $S_X$ (resp.\ $H$) on qubit $n$ (resp.\ qubit $1$) and identity on all other qubits. 
\end{definition}

The circuit diagram for $SCH_n$ is:
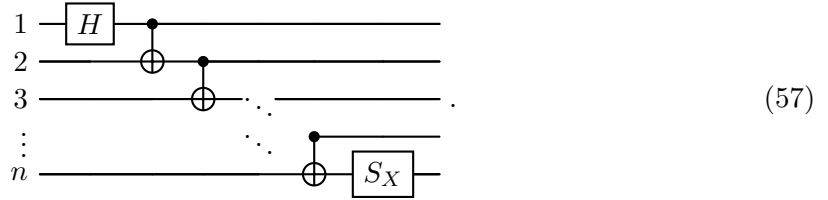
\begin{equation}
  \begin{quantikz}[row sep={0.5cm,between origins}, column sep=0.35cm]
    \lstick{$1$}      & \gate{H} & \ctrl{1} & \qw      & \qw      & \qw &\qw &\qw \\
    \lstick{$2$}      & \qw      & \targ{}  & \ctrl{1} & \qw      & \qw &\qw &\qw \\
    \lstick{$3$}      & \qw      & \qw      & \targ{}  & \ddots   & \qw &\qw &\qw \\
    \lstick{$\vdots$}\setwiretype{n} &       &     &    & \ddots   & \ctrl{1} &\qw&\qw  \\
    \lstick{$n$}      & \qw      & \qw      & \qw      & \qw      & \targ{} & \gate{S_X} &\qw 
  \end{quantikz}.
\end{equation}

\begin{proposition}[$S_X$-CNOT-$H$ string saturates the lower bound]\label{prop:sxcnoth}
  The Clifford unitary $SCH_n$ is Pauli periodic, and saturates the bound in Theorem~\ref{thm:lowerboundqubit}. That is, for every $n> 2$, $SCH_n$ has Pauli periodicity $\lceil \log_2(2n)\rceil$.
\end{proposition}
We defer the verification of its periodicity to Appendix~\ref{app:proof_of_optimality}. As intuition for the large periodicity, note that the CNOT cascade maps a computational basis state to a two-component ``cat'' superposition supported on two Hamming-separated bitstrings, and the final $S_X$ injects a nontrivial phase between these components. Iterating $SCH_n$ repeatedly propagates and mixes these phases across the register, so that the support of a basis state under $(SCH_n)^t$ quickly spreads to an exponentially growing set of computational basis strings. One therefore expects that reaching a Pauli (and hence returning to the identity up to phase) requires a number of iterations that is linear in $n$, consistent with the exact order $2^{\lceil \log_2(2n)\rceil}$ implied by the nilpotency bound. Mathematically, the $H$ and $S_X$ gates spread the Jordan chain in Eq.~\eqref{eq:jordan} to both $X$ and $Z$ operators.

\section{Application: catalyzed logical phase gate}\label{sec:qecapplication}

Controlled unitaries are cornerstones of quantum algorithms; in phase estimation and related routines, coherent control converts an eigenphase of a unitary into a measurable bit string. In our setting, we exploit the same idea in a fault-tolerant context. Specifically, we use jumped Cliffords to prepare logical magic states that are eigenstates of Pauli-periodic Clifford unitaries with eigenvalue $e^{\pi \ii/2^k}$. These eigenstates can then be used as a catalyst to fault-tolerantly implement fine logical phase gates $(\bar Z)^{1/2^{k}}$

Existing measurement-and-postselection protocols typically target eigenstates of Clifford unitaries~\cite{Bravyi2019simulationofquantum,chamberland2020npjQuantumInf,gidney2024,davydova2025universal,chen2025,claes2025cultivating}, and then convert them into useful non-stabilizer resources. Canonical examples are
$\ket{H}\propto \frac{1+H}{2}\ket{+}$ and $\ket{CZ}\propto \frac{1+CZ_{1,2}}{2}\ket{++}$, where one implements the projector by coherently controlling the unitary with an ancilla, measuring the ancilla in the $X$ basis, and post-selecting on the $+1$ outcome.

It is tempting to apply the same approach to Pauli-periodic Cliffords using their control gates, i.e. the jumped Cliffords; however, doing so does not leverage the higher-level nature of jumped Cliffords. The reason is that the corresponding post-selected states already admit exact Clifford+$T$ preparations.
\begin{proposition}\label{prop:periodic-clifford-projector-ct}
  Let $U$ be an $m$-Pauli-periodic Clifford and let $\ket{\psi}$ be a Pauli eigenstate. Then the state $\ket{\psi}$ projected onto the $+1$ eigenspace of $U$, namely $\sum_{r=0}^{2^m-1} U^r \ket{\psi}$, can be prepared exactly using a Clifford+$T$ circuit.
\end{proposition}
\noindent One way to see this is to note that the projector $\sum_{r=0}^{2^m-1} U^r$ can be generated by a circuit containing multiple controlled-$U$ gates followed by Pauli measurements and post-selection; by Proposition~\ref{prop:cliffordT}, controlled Pauli-periodic Cliffords admit exact Clifford+$T$ implementations. Together with the Pauli basis initialization of $\ket{\psi}$, in the ZX-calculus, the corresponding ZX-diagram for the projected state only has spiders whose phases are integer multiples of $\pi/4$~\cite{Backens2014,JeandelPerdrixVilmart2019}.

To implement finer phases, we should use jumped Cliffords in a way that exploits their place in the high level of the Clifford hierarchy. The key idea is to instead prepare an eigenstate of the Pauli-periodic Clifford $U$ with eigenvalue $e^{\pi \ii/2^k}$. Such an eigenvalue of $U$ could exist for an $m$-Pauli-periodic Clifford with Pauli periodicity of at least $m=k-1$, in which case the Pauli operator obtained via squaring the unitary $U^{2^m}=P$ contains a phase $\ii$\footnote{We note that in this case the level of the controlled gate $CU$ matches the level of the phase gate $Z^{1/2^k}$: both lie in the $(k+1)$-th level from Proposition~\ref{prop:CH-basic}(6) and Theorem~\ref{thm:controlled-jump-criterion}. Therefore, the effect of the controlled jump has been fully exploited. If $P=U^{2^m}$ do not contain phase $i$, but at least is a non-trivial Pauli string, we will ``waste'' the jumped level by $1$.}. This is known as the catalyst state in the literature~\cite{campbell2011catalysis,Gidney2019efficientmagicstate,Beverland2020lower}. We denote this eigenstate by $\ket{\psi_k}$, which satisfies
\begin{equation}
  U\ket{\psi_k}=e^{\pi \ii/2^k}\ket{\psi_k}.
\end{equation}
Once $\ket{\psi_k}$ is available, a single application of the jumped Clifford $CU$ ``kicks back'' this eigenphase onto the control qubit, thereby producing the standard single-qubit phase gate $Z^{1/2^k}$:
\begin{equation}
  CU\,\ket{\phi}\otimes\ket{\psi_k}=\phi_0\ket{0}\otimes\ket{\psi_k}+\phi_1\ket{1}\otimes U\ket{\psi_k}=\left(\phi_0\ket{0}+e^{\frac{\pi \ii}{2^k}}\phi_1\ket{1}\right)\otimes\ket{\psi_k}=Z^{1/2^k}\ket{\phi}\otimes \ket{\psi_k}.
\end{equation}
Here $\ket{\phi}=\phi_0\ket{0}+\phi_1\ket{1}$ is an arbitrary single qubit state.

In a fault-tolerant implementation, this reduces the problem of fault-tolerantly implement $\overline{Z^{1/2^k}}$ to preparing the logical catalyst eigenstate $\ket{\overline{\psi_k}}$ and applying $CU$ once at the logical level. Importantly, Proposition~\ref{prop:cliffordT} ensures that the logical gate $CU$ can be implemented exactly using Clifford+$T$. Thus, if one works with a quantum error-correcting code that admits transversal implementations of both $T$ and the chosen periodic Clifford $U$, the entire routine can be executed transversally without any additional non-stabilizer resource states beyond those used for transversal $T$. For instance, one may take $U$ to be a CNOT string, like the constant depth one in Eq.~\eqref{eq:brickworkcnot}, and use the 3D color code~\cite{bombin2007exact,Kubica2015unfolding,zhu2025nonclifford}, which supports transversal implementations of both $T$ and CNOT, since it is a Calderbank-Shor-Steane (CSS) code~\cite{calderbank1996good,steane1996multiple}.

We now describe a logical-level circuit for preparing $\ket{\overline{\psi_k}}$. 
Suppose $U$ is supported on $n$ physical qubits, and that the chosen code admits a transversal implementation of $U$ across $n$ code blocks. Then the resulting logical operation factorizes as
\begin{equation}
  \bar U=\otimes_{p=1}^{l}U_p,
\end{equation}
where $l$ is the number of disjoint transversal parties, each contains $n_p=O(1)$ qubits, and $U_p$ denotes the restriction of $U$ to that party. Consequently, a controlled logical operation between a single ancilla and the encoded data can be implemented as a product of smaller controlled gates,
\begin{equation}
  C(\bar U)=\prod_{p=1}^l C(U_p),
\end{equation}
as illustrated by the circuit representation
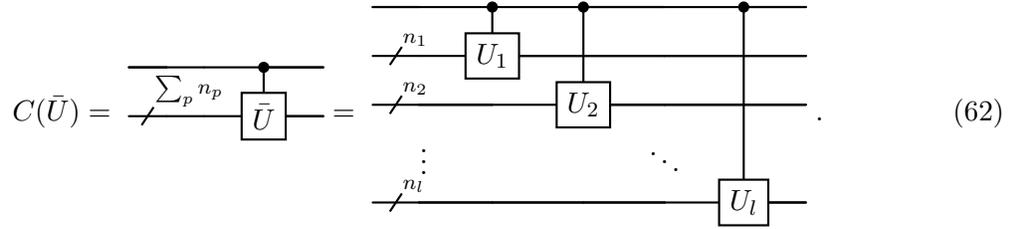
\begin{equation}\label{eq:cubar}
  C(\bar U)=\begin{quantikz}[row sep={0.65cm,between origins}, column sep=0.5cm]
    & \qw & &\ctrl{1} &  \qw \\
    & \qwbundle{\sum_p n_p} & &\gate{\bar U} & \qw \\
    \end{quantikz}
    =
    \begin{quantikz}[row sep={0.65cm,between origins}, column sep=0.5cm]
    & \qw & \ctrl{1} & \ctrl{2} & \qw      & \ctrl{4} & \qw \\
    & \qwbundle{n_1} & \gate{U_1} & \qw      & \qw      & \qw      & \qw \\
    & \qwbundle{n_2} & \qw      & \gate{U_2} & \qw      & \qw      & \qw \\
    \setwiretype{n}& \ \vdots &      &     & \ddots   &       &  \\
    & \qwbundle{n_l} & \qw      & \qw      & \qw      & \gate{U_l} & \qw
  \end{quantikz}.
\end{equation}
With this decomposition in hand, we prepare $\ket{\overline{\psi_k}}$ via quantum phase estimation (QPE) on $\bar U$ using physical ancillas. Concretely, we use $k+1$ ancilla qubits to resolve phases in multiples of $2\pi/2^{k+1}$ (since the target eigenphase is $\pi/2^k$), apply controlled powers of $\bar U$, and then apply the inverse quantum Fourier transform (QFT) and measure the ancillas:
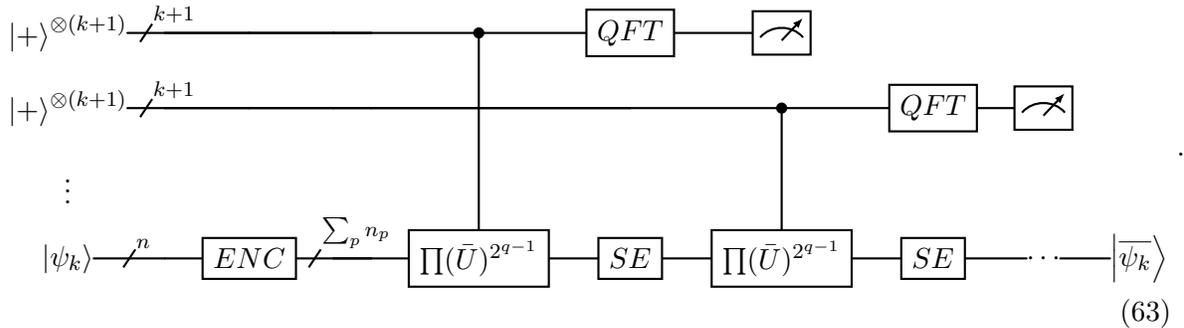
\begin{equation}
  \begin{quantikz}[row sep={1cm,between origins}, column sep=0.5cm]
  \ket{+}^{\otimes (k+1)} & \qwbundle{k+1} & \qw &\qw &\qw & \ctrl{3} &\gate{QFT} & \meter{}\\
  \ket{+}^{\otimes (k+1)} & \qwbundle{k+1} & \qw &\qw &\qw &\qw &\qw & \ctrl{2} &\gate{QFT} & \meter{}\\
   \setwiretype{n}\vdots& &      &     &    &       &  \\
  \ket{\psi_k} & \qwbundle{n} & \gate{ENC} & \qwbundle{\sum_p n_p} &\qw & \gate{\prod(\bar U)^{2^{q-1}}} & \gate{SE} & \gate{\prod(\bar U)^{2^{q-1}}} & \gate{SE} &\cdots & \ket{\overline{\psi_k}}
  \end{quantikz}.
\end{equation}
Here we first prepare the physical eigenstate $\ket{\psi_k}$ of $U$ on $n$ qubits. This can be done using physical qubit-level QPE. We then encode it (ENC) into $n$ logical qubits, yielding an initial, noisy realization of $\ket{\overline{\psi_k}}$. We then run QPE between the $k+1$ ancillas and the encoded register by applying $C((\bar U)^{2^{q-1}})$ controlled by the $k$th ancilla (for $q=1,\dots,k+1$), and post-select the measurement outcome $(1)_2$ (i.e. all the digits are 0 except for the first one), which corresponds to the desired eigenvalue $e^{\frac{\ii\pi}{2^k}}$ of $U$. Between these controlled-power steps, we interleave syndrome extraction (SE) rounds, so that repeating the routine improves the fidelity of the resulting logical eigenstate. Adding the steps of growing the code distance, the entire protocol can be potentially developed into a cultivation scheme of the catalyst state $\ket{\psi_m}$.

In each round of the logical measurement circuit in Eq.~\eqref{eq:cubar}, the depth scales linearly with the number of transveral parties $l$, since $C(\bar U)$ requires sequential action of physical $CU$ gates on each transversal party, and each physical $U_p$ can be performed in constant depth, if one uses constant depth realization of periodic Cliffords, such as the brickwork CNOT circuit in Eq.~\eqref{eq:brickworkcnot}. The depth of physical $CU$ gates can be further reduced to a constant if one uses a fan-out circuit on between the ancilla and a set of $l$ additional ancillas before the $CU$ gates, similar to the idea in Ref.~\cite{kim2025any}. The depth of the (non-fault-tolerant) QFT part of the circuit scales linearly with $k$.

We point out that the controlled logical gates in the logical QPE, $C((\bar U)^{2^{q-1}})$, and the phase gates in the physical QFT, all lie in definite levels the Clifford hierarchy for every $q$, thanks to Proposition~\ref{prop:tensor-product-periodicity}. 
In addition, our ability to pinpoint the level of the hierarchy for the entire logical circuit suggests our protocol for catalyst state preparation may be amenable to more systematic analysis of Pauli error propagation which aids to the possibility of classical simulation and benchmarking. We leave the fault-tolerance analysis and the development of simulation technique to future work. 

\section{Conclusion and Open Problems}
\label{sec:conclusion}
\label{sec:jumpfromnonclif}% Retained for backward compatibility with in-text cross-references
% NOTE: The label sec:jumpfromnonclif is kept for compatibility with the cross-reference at line 443,
% but sec:conclusion is the canonical label for this section.

This work identifies a simple mechanism, coherent control of periodic Clifford gates, that generates explicit, exactly-synthesizable families of high-level Clifford-hierarchy operations.
Our first contribution is conceptual: we introduced \emph{Pauli periodicity} as an invariant tailored to controlled constructions, capturing the precise point at which repeated squaring of a Clifford collapses to a Pauli.
Our main structural result, Theorem~\ref{thm:controlled-jump-criterion}, turns this invariant into a sharp controlled-jump rule: for a Pauli-periodic Clifford $U$, the hierarchy level of $CU$ is determined exactly by the Pauli periodicity of $U$.
We also showed that the resulting jumped Cliffords enjoy useful closure and compilation properties, such as stability under inversion and under tensor-product targets, and exact Clifford+$T$ synthesis.
Our second contribution is quantitative: Theorem~\ref{thm:lowerboundqubit} bounds Pauli periodicity as a function of the number of qubits, and together with explicit saturating families it yields a tight resource tradeoff. Namely, although controlled Cliffords can reach arbitrarily high hierarchy levels in principle, a jump to level $k$ from Clifford targets requires a number of target qubits that grows exponentially with $k$.
Finally, we connected the algebraic picture to fault tolerance by proposing a catalyst-state preparation routine and a catalyzed logical phase-gate protocol based on a single jumped Clifford.

On the level of mathematics, our results raise several natural open problems. A first direction is to go beyond Clifford inputs and study the \emph{controlled-jumping power} of unitaries already in the third and higher levels of the hierarchy. In particular:
\begin{open}
  Among all the $n$-qubit unitaries $U\in \mathcal{C}^{(n)}_k$ in the $k$-th level of the Clifford hierarchy with Pauli periodicity $m$ (i.e. $U^{2^m}\in \mathcal{P}_n$), what is the highest level of the Clifford hierarchy that $CU$ can reach?
\end{open}
\noindent This problem is subtle because the outcome can depend simultaneously on three parameters: the periodicity $m$, the number of target qubits $n$, and the input level $k$.
For Clifford inputs, Theorem~\ref{thm:lowerboundqubit} tightly couples $m$ to $n$, but in higher levels this coupling can disappear: for example, the single-qubit phase gate $Z^{1/2^{k-1}}\in \mathcal{C}^{(1)}_k$ (Proposition~\ref{prop:CH-basic}(6)) has Pauli periodicity $k-1$ independent of $n$.
At the same time, large periodicity does \emph{not} automatically translate into a large controlled jump: Ref.~\cite{CuiGottesmanKrishna2017} shows that $CZ^{1/2^{k-1}}\in \mathcal{C}^{(2)}_{k+1}$, i.e., adding a control increases the level by only one.
Understanding when coherent control produces a \emph{genuinely super-constant} level increase for non-Clifford inputs would reveal new structure in $\mathcal{C}_k$ and could potentially bypass the exponential qubit overhead inherent to Clifford targets.

Independently of controlled jumps, the relationship between Pauli periodicity and the number of qubits is interesting in its own right. More explicitly:
\begin{open}
  Given a Pauli-periodic unitary $U\in \mathcal{C}^{(n)}_k$, what is the largest Pauli periodicity $m_\mathrm{max}$ that can be achieved by $U$? In particular, for a fixed high-enough level $k$, is it possible to reach $m_\mathrm{max}\sim \mathrm{poly}(n)$ instead of $\log (n)$? 
\end{open}
\noindent For the problem to be meaningful, we must have $m\geq k$, since we always have a $(k-1)$-Pauli-periodic unitary in the $k$-th level, namely $Z^{1/2^{k-1}}$. This question is highly meaningful given the fact that if one drops the Pauli-periodicity constraint and instead asks for the maximal order of an $n$-qubit unitary within a given level of the hierarchy, one can obtain the maximum period of $2^n-1$ (which is \emph{not} a power of 2). A famous example is the linear-feedback shift-register (LFSR) construction~\cite{tausworthe1965random,massey1969shift}, which can be implemented by Clifford circuits in quantum computation~\cite{kim2025catalytic,ippoliti2025infinite}. In contrast, the open problem poses a much stronger requirement that some $2$-power of $U$ collapses to a Pauli. In the context of controlled jump, if $m_\mathrm{max}\sim \mathrm{poly}(n)$ for a certain level $k$, the qubit resource needed for controlled jump, and subsequently the protocol for logical phase gate, could be dramatically reduced.

The discussion in Sec.~\ref{sec:cclifperm} raises a related set of questions about the complexity of permutation gates within the Clifford hierarchy. Proposition~\ref{prop:jumped-clifford-perm-lowdeg} shows that the Clifford-dressed jumped permutation gates admit quadratic coordinate maps and yet can occupy arbitrarily high hierarchy levels given large enough $n$. However, the quadratic structure realized there is highly restricted: all genuinely quadratic terms arise only from products between a single control bit and affine-linear functions of the target bits. This motivates asking for a principled, algorithmic classification of general quadratic permutations.
\begin{open}
  Let $U_\pi$ be an $n$-qubit permutation gate implementing a bijection $\pi:\{0,1\}^n\to\{0,1\}^n$. Assume that both $\pi$ and/or $\pi^{-1}$ admit coordinatewise representations by polynomials over $\mathbb{F}_2$ of total degree at most $2$ (or in general a fixed degree $l$). Given such a polynomial description of $\pi$, is there an efficient algorithm that computes the smallest $k$ such that $U_\pi\in\mathcal{C}^{(n)}_k$? 
\end{open}

On the level of applications, we hope that Pauli-periodic Cliffords and their controlled jumps can lead to more efficient preparation and simulation of resource states that can be used to implement high-level logical gates, such as the catalyzed phase-gate protocol in Sec.~\ref{sec:qecapplication}. It is instructive to contrast our setting with approaches such as Ref.~\cite{kim2025catalytic}, which realize a fine phase $2\pi/(2^k-1)$ using eigenstates of periodic Cliffords on only $k$ qubits, rather than the $n\sim e^k$ target qubits required to implement a level-$k$ controlled jump using Clifford targets in our protocol. The tradeoff there is that the controlled unitaries in their case lie completely outside the Clifford hierarchy according to Theorem~\ref{thm:AW-necessary}. This can make classical simulation and fault-tolerance analysis substantially more challenging for the fault-tolerant version of such an algorithm. 

\paragraph{Acknowledgement} YX acknowledges support by the NSF through the grant OAC-2118310. XW acknowledges support by the U.S. Department of Energy through Award Number DE-SC0023905.

\bibliographystyle{quantum}
\bibliography{main}

\appendix
\section{Proof of Proposition~\ref{prop:sxcnoth}}\label{app:proof_of_optimality}

From Lemma~\ref{lem:binary-matrix-clifford}, to prove the Pauli periodicity of $SCH_n$ is $\lceil \log_2(2n)\rceil$, we show that $N_n:=M_n-I_{2n}$, where $M_n$ is the binary symplectic matrix for $SCH_n$, has nilpotency index $2n$. That is, it satisfies $N_n^{2n}=0$ but $N_n^{r}\neq 0$ for all $0\leq r\leq 2n-1$.

The matrix $M_n$ has the following block form:
\begin{equation}
  M_n=\begin{pmatrix}A_n & B_n\\ C_n & D_n\end{pmatrix}\in \mathrm{Sp}(2n,\mathbb{F}_2),
\end{equation}
where the $n\times n$ blocks $A_n,B_n,C_n,D_n$ are given entrywise (for $1\le i,j\le n$) by
\begin{align}
  (A_n)_{ij} &=
  \begin{cases}
    1, & i\ge j \text{ and } j>1,\\
    0, & \text{otherwise},
  \end{cases}
  \label{eq:Ap-def}\\
  (B_n)_{ij} &=
  \begin{cases}
    1, & j=1 \text{ or } (i=j=n),\\
    0, & \text{otherwise},
  \end{cases}
  \label{eq:Bp-def}\\
  (C_n)_{ij} &=
  \begin{cases}
    1, & i=j=1,\\
    0, & \text{otherwise},
  \end{cases}
  \label{eq:Cp-def}\\
  (D_n)_{ij} &=
  \begin{cases}
    1, & (i=j \text{ and } i>1)\text{ or } (j=i+1),\\
    0, & \text{otherwise}.
  \end{cases}
  \label{eq:Dp-def}
\end{align}

We now compute the action of $N_n$ on the symplectic basis. Let $\{e_j\}_{j=1}^n$ be the standard basis of $\mathbb{F}_2^n$, i.e. $(e_j)_{k}=\delta_{j,k}$.
We write elements of $(\mathbb{F}_2)^{2n}$ as column vectors $(x,z)$ with $x,z\in\mathbb{F}_2^n$, and define
\[
e_j^X := (e_j,0)^T,
\qquad
e_j^Z := (0,e_j)^T,
\qquad (1\le j\le n).
\]
Using the action of $M_n$, the bitstrings are transformed as $(x',z')^T=(A_nx+B_nz,\; C_nx+D_nz)^T$. Therefore, we have the following identities in $(\mathbb{F}_2)^{2n}$.
For the $X$-basis vectors,
\begin{align}
N_n e_1^X
&=(A_n-I)e_1 \;+\; (C_n)e_1
\;=\; e_1^X+e_1^Z,
\label{eq:N-x1}\\
N_n e_j^X
&=((A_n-I)e_j,\; C_ne_j)
\;=\;\Bigl(\sum_{i=j+1}^n e_i,\;0\Bigr)
\;=\;\sum_{i=j+1}^n e_i^X,
\qquad 2\le j\le n-1,
\label{eq:N-xj}\\
N_n e_n^X
&=0.
\label{eq:N-xn}
\end{align}
For the $Z$-basis vectors,
\begin{align}
N_n e_1^Z
&=(B_ne_1,\; (D_n-I)e_1)
\;=\;\Bigl(\sum_{i=1}^n e_i,\; e_1\Bigr)
\;=\; e_1^Z+\sum_{i=1}^n e_i^X,
\label{eq:N-z1}\\
N_n e_j^Z
&=(B_ne_j,\; (D_n-I)e_j)
\;=\;(0,\; e_{j-1})
\;=\; e_{j-1}^Z,
\qquad 2\le j\le n-1,
\label{eq:N-zj}\\
N_n e_n^Z
&=(B_ne_n,\; (D_n-I)e_n)
\;=\;(e_n,\; e_{n-1})
\;=\; e_n^X+e_{n-1}^Z.
\label{eq:N-zn}
\end{align}

We now show that, starting from $e_n^Z$, the repeated action of $N_n$ generates a length-$2n$ Jordan chain.
From \eqref{eq:N-zn} and \eqref{eq:N-xn} we have
\[
N_n^k e_n^Z
= N_n^{k-1}(e_{n-1}^Z+e_n^X)
= N_n^{k-1}e_{n-1}^Z
\qquad (k\ge 2).
\]
Iterating \eqref{eq:N-zj} gives
\begin{equation}
\label{eq:zn-to-z1}
N_n^{\,n-1} e_n^Z \;=\; e_1^Z.
\end{equation}
Applying \eqref{eq:N-z1} then yields
\begin{equation}
\label{eq:zn-to-z1-plusX}
N_n^{\,n} e_n^Z
\;=\;
N_n e_1^Z
\;=\;
e_1^Z+\sum_{i=1}^n e_i^X.
\end{equation}
Applying $N_n$ once more, and using \eqref{eq:N-z1} together with \eqref{eq:N-x1}--\eqref{eq:N-xn}, we find
\begin{align*}
N_n^{\,n+1} e_n^Z
&= N_n\Bigl(e_1^Z+\sum_{i=1}^n e_i^X\Bigr) \\
&= \Bigl(e_1^Z+\sum_{i=1}^n e_i^X\Bigr)
+ \Bigl(e_1^X+e_1^Z\Bigr)
+ \sum_{j=2}^{n-1}\ \sum_{i=j+1}^n e_i^X \\
&= \sum_{i=1}^n e_i^X + e_1^X + \sum_{i=3}^n (i-2)\,e_i^X.
\end{align*}
Working over $\mathbb{F}_2$, the coefficient $(i-2)$ equals $1$ if and only if $i$ is odd.
Hence the odd-$i$ terms cancel against $\sum_{i=1}^n e_i^X$, and we obtain the clean form
\begin{equation}
\label{eq:even-X}
N_n^{\,n+1} e_n^Z
\;=\;
\sum_{\substack{2\le i\le n\\ i\ \mathrm{even}}} e_i^X
\;=:\; v_{\mathrm{even}}.
\end{equation}

We now analyze powers of $N_n$ on the $X$-subspace $\mathrm{span}\{e_2^X,\dots,e_n^X\}$.
Define $y_k:=e_{n-k+1}^X$ for $1\le k\le n-1$, so $y_1=e_n^X$ and $y_{n-1}=e_2^X$.
From \eqref{eq:N-xj}--\eqref{eq:N-xn} we have
\[
N_n y_1=0,
\qquad
N_n y_k = y_1+y_2+\cdots+y_{k-1}\quad (k\ge 2).
\]
We claim that for every $k\ge 1$,
\begin{equation}
\label{eq:Jordan-claim}
N_n^{\,k-1}y_k = y_1
\qquad\text{and}\qquad
N_n^{\,k}y_k=0.
\end{equation}
The statement is immediate for $k=1$.
For $k\ge 2$, note that $N_n y_k = y_{k-1} + N_n y_{k-1}$, hence
\[
N_n^{\,k-1}y_k
= N_n^{\,k-2}(N_n y_k)
= N_n^{\,k-2}y_{k-1} + N_n^{\,k-1}y_{k-1}
= y_1 + 0,
\]
where we used \eqref{eq:Jordan-claim} for $k-1$.
Applying $N_n$ once more gives $N_n^{\,k}y_k=N_n y_1=0$, proving \eqref{eq:Jordan-claim}.

Taking $k=n-1$ in \eqref{eq:Jordan-claim} yields
\begin{equation}
\label{eq:Nn-2-x2}
N_n^{\,n-2} e_2^X \;=\; e_n^X,
\qquad
N_n^{\,n-1} e_2^X \;=\; 0.
\end{equation}
Moreover, for any $j\ge 3$, repeated use of \eqref{eq:N-xj} shows that $N_n^{\,n-2}e_j^X=0$ (since the $X$-support strictly shifts to the right at each application).
Because $v_{\mathrm{even}}$ in \eqref{eq:even-X} contains $e_2^X$ and only $e_j^X$ with $j\ge 3$ otherwise, we conclude
\begin{equation}
\label{eq:Nn-2-even}
N_n^{\,n-2} v_{\mathrm{even}} \;=\; e_n^X\neq 0,
\qquad
N_n^{\,n-1} v_{\mathrm{even}} \;=\; 0.
\end{equation}

Combining \eqref{eq:even-X} and \eqref{eq:Nn-2-even} gives
\[
N_n^{\,2n-1} e_n^Z
= N_n^{\,n-2}\bigl(N_n^{\,n+1}e_n^Z\bigr)
= N_n^{\,n-2}v_{\mathrm{even}}
= e_n^X
\neq 0,
\]
while
\[
N_n^{\,2n} e_n^Z
= N_n^{\,n-1}\bigl(N_n^{\,n+1}e_n^Z\bigr)
= N_n^{\,n-1}v_{\mathrm{even}}
= 0.
\]
Therefore, the nilpotency index of $N_n$ is exactly $2n$.

\end{document}